\documentclass[journal]{IEEEtran}
\usepackage{amsfonts}
\usepackage{amssymb}
\usepackage{amsthm}
\usepackage{amsmath,amsfonts,amssymb}
\usepackage[dvips]{graphicx}
\usepackage{verbatim}
\usepackage{slashbox}
\usepackage{setspace}
\usepackage{bm}
\usepackage{algorithmic} 
\usepackage[ruled,vlined]{algorithm2e}
\usepackage{cite}
\usepackage{changepage}

\usepackage[bookmarks=false]{}
\newtheorem{theorem}{Theorem}

\newtheorem{corollary}{Corollary}

\newcommand{\figwidth}{9}
\IEEEoverridecommandlockouts

\begin{document}

\title{Hierarchical Codebook Design for Beamforming Training in Millimeter-Wave Communication}

\author{Zhenyu Xiao,~\IEEEmembership{Member,~IEEE,}
Tong He,
        Pengfei Xia,~\IEEEmembership{Senior Member,~IEEE,}
and Xiang-Gen Xia,~\IEEEmembership{Fellow,~IEEE}
\thanks{This work was partially supported by the National Natural Science Foundation of China (NSFC) under grant Nos. 61571025, 91338106, 91538204, and 61231013, National Basic Research Program of China under grant No.2011CB707000, and Foundation for Innovative Research Groups of the National Natural Science Foundation of China under grant No. 61221061.}
\thanks{Z. Xiao and T. He are with the School of
Electronic and Information Engineering, Beijing Key Laboratory for Network-based Cooperative Air Traffic Management, and Beijing Laboratory for General Aviation Technology, Beihang University, Beijing 100191, P.R. China.}
\thanks{P. Xia is with the School of Electronics and Information Engineering, Tongji University, Shanghai, P.R. China.}
\thanks{X.-G. Xia is with the Department of Electrical and Computer Engineering, University of Delaware, Newark, DE 19716, USA.}

\thanks{Corresponding Author: Dr. Z. Xiao with Email: xiaozy@buaa.edu.cn.}
}
\markboth{IEEE TRANSACTIONS ON WIRELESS COMMUNICATIONS, VOL. X, NO.
x, xxx 20xx}{Shell \MakeLowercase{\textit{et al.}}}

\maketitle
\begin{abstract}
In millimeter-wave communication, large antenna arrays are required to achieve high power gain by steering towards each other with narrow beams, which poses the problem to efficiently search the best beam direction in the angle domain at both Tx and Rx sides. As the exhaustive search is time consuming, hierarchical search has been widely accepted to reduce the complexity, and its performance is highly dependent on the codebook design. In this paper, we propose two basic criteria for the hierarchical codebook design, and devise an efficient hierarchical codebook by jointly exploiting sub-array and deactivation (turning-off) antenna processing techniques, where closed-form expressions are provided to generate the codebook. Performance evaluations are conducted under different system and channel models. Results show superiority of the proposed codebook over the existing alternatives.
\end{abstract}


\section{Introduction}
\IEEEPARstart{M}{illimeter-wave} (mmWave) communication is a promising technology for next-generation wireless communication owing to its abundant frequency spectrum resource, which promises a much higher capacity than the existing wireless local area networks (WLANs) and the current cellular mobile communication. In fact, mmWave communication has received increasing attentions as an important candidate technology in both the next-generation WLANs \cite{Rapp_2010_60GHz_general,wang_2011_MMWCS,Park_2010_11ad,Xia_2011_60GHz_Tech,xiaozhenyu2013div,xia_2008_prac_ante_traning,xia_2008_multi_stage} and mobile cellular communication \cite{khan_2011,alkhateeb2014mimo,choi2014coding,han2015large,roh2014millimeter,sun2014mimo,niu2015survey,wang2014tens,wang2015multi}.
A fundamental challenge to mmWave communication is the extremely high path loss, thanks to the very high carrier frequency on the order of 30-60 GHz. To bridge this significant link budget gap, joint Tx/Rx beamforming is usually required to bring large antenna array gains, which typically requires a large Tx/Rx antenna array size (e.g., an antenna array size of 36 is used in \cite{Xia_2011_60GHz_Tech}.). Fortunately, thanks to the small wavelength on the mmWave frequency, large antenna arrays are possible to be packed into a small area.

 In the mmWave domain, the high power consumption of mixed signal components, as well as expensive radio-frequency (RF) chains, make it difficult, if not impossible, to realize digital baseband beamforming as used in the conventional multiple-input multiple-output (MIMO) systems. Instead, analog beamforming is usually preferred, where all the antennas share a single RF chain and have constant-amplitude (CA) constraint on their weights \cite{Xia_2011_60GHz_Tech, xia_2008_prac_ante_traning,wang_2009_beam_codebook,wang_2009_beam_codebook_vtc}. Meanwhile, a hybrid analog/digital precoding structure was also proposed to realize multi-stream/multi-user transmission \cite{alkhateeb2014mimo,roh2014millimeter,sun2014mimo}, where a small number of RF chains are tied to a large antenna array. No matter whether the analog beamforming or the hybrid precoding structure is exploited, entry-wise estimation of channel state information (CSI) is time costly due to large-size antenna arrays, and a more efficient antenna training algorithm is needed.

 For the hybrid precoding structure, as the mmWave channel is generally sparse in the angle domain, different compressed sensing (CS) based channel estimation methods were proposed to estimate the steering angles of multipath components (MPCs) \cite{alkhateeb2014channel,alkhateeb2014mimo,alkhateeb2015compressed,peng2015enhanced,wang2015multi}, where \cite{alkhateeb2014channel} is for point-to-point multi-stream transmission, \cite{alkhateeb2015compressed} is for multi-user transmission, while \cite{peng2015enhanced}, based on a presentation of antenna array with
virtual elements, further enhances the channel estimation over \cite{alkhateeb2014channel}. For the analog beamforming structure, there in general exist two different approaches. In \cite{Xia_2011_60GHz_Tech,xia_2008_prac_ante_traning,xia_2008_multi_stage,Xiaozy2014BeamTrain}, an iterative beamforming training approach is adopted, in which the beamforming vector on one side (transmitter or receiver) is alternatively optimized by fixing the beamforming vector on the other side, and the alternation is repeated iteratively to improve the beamforming gain one iteration upon another. On the other hand, in \cite{wang_2009_beam_codebook,wang_2009_beam_codebook_vtc,kokshoorn2015fast}, a switched beamforming approach is adopted, where the beam search space (at the transmitter and receiver side, respectively) is represented by a codebook containing multiple codewords, and the best transmit/receive beams are found by searching through their respective codebooks. Both approaches have their own merit and may be useful in different applications.



In this paper, we focus on switched beamforming for one-stream transmissions. This is motivated by the fact that single-stream beamforming is actually capacity achieving in the very-low SNR case \cite{TseFundaWC}. Furthermore, single stream beamforming can be readily extended to the more complicated hybrid precoding case \cite{alkhateeb2014channel}. For switched beamforming, an exhaustive search algorithm may be used, which sequentially tests all the beam directions in the angle domain and finds the best pair of transmit/receive beamforming codewords. This is conceptually straightforward. However, the overall search time is prohibitive, because the number of candidate beam directions is usually large for mmWave communication. To improve the search efficiency, a hierarchy of codebooks may be defined \cite{wang_2009_beam_codebook,wang_2009_beam_codebook_vtc,chen2011multi,Park_2010_11ad,hur2013millimeter,he2015suboptimal}. For example, a coarse codebook may be defined with a small number of coarse/low-resolution beams (or sectors) covering the intended angle range, while a fine codebook may be defined with a large number of fine/high-resolution beams covering the same intended angle range, and that a coarse beam may have the same/similar coverage as that of multiple fine beams together. A divide-and-conquer search may then be carried out across the hierarchy of codebooks, by finding the best sector (or coarse beam) first on the low-resolution codebook level, and then the best fine beam on the high-resolution codebook level, with the best high-resolution beam contained by the best low-resolution beam.


Performances of the switched beamforming schemes, including the search time and success rate, are highly dependent on the hierarchical codebook design. In \cite{wang_2009_beam_codebook,wang_2009_beam_codebook_vtc}, although wider beams were proposed to speed up beam search, design approaches to broaden the beams were not studied. In \cite{chen2011multi}, codewords with wider beams were generated by summing two codewords with narrower beams, but the CA constraint was no longer satisfied. In \cite{hur2013millimeter}, a sub-array method was proposed to broaden the beam width by pointing the sub-array beams in slightly-gapped directions. However, a full hierarchical codebook was not explicitly designed therein, and this approach may be not feasible to design very wide or even omni-directional beams. In \cite{alkhateeb2014channel}, the hybrid precoding structure was adopted to shape wider beams by exploiting the sparse reconstruction approach, but high-quality wide beams can be shaped only when the number of RF chains is large enough and deep sinks within the angle range appear otherwise. In \cite{he2015suboptimal}, a binary-tree structured hierarchical codebook was designed by using antenna deactivation, where wider beams were generated by turning off part of the antennas. A complete codebook was designed with closed-form expressions provided therein. However, the number of active antennas is too small for very wide or omni-directional beams, which may limit its application in mmWave communication, where per-antenna transmission power is limited.

In this paper, we first propose two basic criteria for arbitrary hierarchical codebook designs, and then devise an efficient hierarchical codebook by jointly exploiting sub-array and deactivation (turning-off) antenna processing techniques. Closed-form expressions are provided to generate the codebook. In the proposed approach, the beams of the sub-arrays steer towards widely-gapped directions to broaden beams, which is essentially different from \cite{hur2013millimeter}, and the deactivation operates on the sub-arrays instead of individual antennas like that in \cite{he2015suboptimal}. To the best of our knowledge, this is the first to propose these two criteria and the joint sub-array and deactivation codebook design. Performance evaluations are conducted under both line-of-sight (LOS) and non-LOS (NLOS) channels, as well as with both total and per-antenna transmission power models. Results show superiority of the proposed codebook over the existing alternatives, especially when the per-antenna transmit power is constrained.

The rest of this paper is as follows. In Section II, the system and channel models are introduced. In Section III, the hierarchical codebook design is presented. In Section IV, performance evaluation is conducted. The conclusion is drawn lastly in Section V.

Symbol Notations: $a$, $\mathbf{a}$, $\mathbf{A}$, and $\mathcal{A}$ denote a scalar variable, a vector, a matrix, and a set, respectively. $(\cdot)^{\rm{*}}$, $(\cdot)^{\rm{T}}$ and $(\cdot)^{\rm{H}}$ denote conjugate, transpose and conjugate transpose, respectively. $\mathbb{E}(\cdot)$ denotes expectation operation. $[\mathbf{a}]_i$ and $[\mathbf{A}]_{ij}$ denote the $i$-th entry of $\mathbf{a}$ and the $i$-row and $j$-column entry of $\mathbf{A}$, respectively. $[\mathbf{a}]_{i:j}$ denotes a vector with entries being the $i$-th to $j$-th entries of $[\mathbf{a}]$. $|\cdot|$ and $\|\cdot\|$ denote the absolute value and two-norm, respectively.

\section{System and Channel Models}
\subsection{System Model}
Without loss of generality, we consider an mmWave communication system with
half-wave spaced uniform linear arrays (ULAs) of $N_{\rm{T}}$ and $N_{\rm{R}}$ elements equipped
at the transmitter and receiver, respectively \cite{wang_2009_beam_codebook,xiaozhenyu2013div,xiaozhenyuGC2013,xiao2015Iterative}, as shown in Fig. \ref{fig:system}, where a single RF chain is tied to the ULA at both the transmitter and receiver, and thus the analog beamforming structure is adopted. At the transmitter, each antenna branch has a phase shifter and power amplifier (AP) to drive the antenna, while at the receiver, each antenna branch has a low-noise amplifier (LNA) to amplify the signal and a phase shifter. It is noted that as the analog beamforming can be seen as one of the branches of the hybrid precoding structure, the proposed criteria and codebook design can be naturally used by the hybrid precoding structure, which will be shown in Section III-D. Additionally, in our system each antenna branch can be deactivated or turned off, i.e., there is a switch in each antenna branch at both sides although not depicted in the figure. Generally, all the PAs, as well as the LNAs, have the same scaling factor if activated. Thus, each element of the antenna weight vectors (AWVs) at the both sides either has a constant amplitude or is zero.

\begin{figure}[t]
\begin{center}
  \includegraphics[width=\figwidth cm]{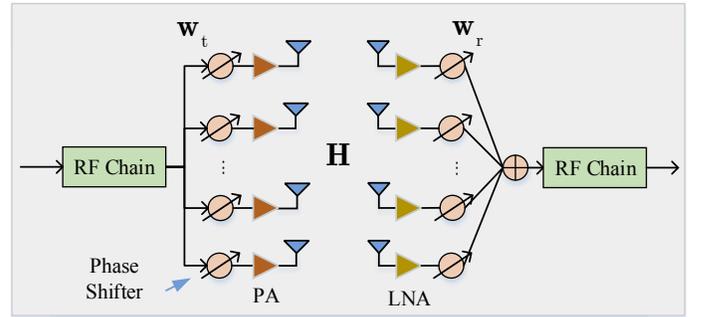}
  \caption{Illustration of the system.}
  \label{fig:system}
\end{center}
\end{figure}

Letting $s$ denote the transmitted symbol with unit power, the received signal is
\begin{equation} \label{eq_signal_totalp}
y = \sqrt{P_{\rm{tot}}} {\bf{w}}_{\rm{R}}^{\rm{H}}{\bf{H}}{{\bf{w}}_{\rm{T}}}s + {\bf{w}}_{\rm{R}}^{\rm{H}}{\bf{n}},
\end{equation}
where $P_{\rm{tot}}$ is the \emph{total transmission power} of all the active antennas, ${\bf{w}}_{\rm{T}}$ and ${\bf{w}}_{\rm{R}}$ are the transmit and receive AWVs, respectively, ${\bf{H}}$ is the channel matrix, $\bf{n}$ is the Gaussian noise vector with power $N_0$, i.e., $\mathbb{E}({\bf{n}}{\bf{n}}^{\rm{H}})={N_0\bf{I}}$. Let ${\cal W}(N)$ denote a set of vectors with $N$ entries as shown in \eqref{eq_setW}, where $\nu$ is a normalization factor such that all the vectors have \emph{unit power}. We can find that each entry of an arbitrary vector in ${\cal W}(N)$ has either an amplitude $\nu$ (activated) or is 0 (deactivated). Consequently, we have ${{\bf{w}}_{\rm{T}}} \in {\cal W}({N_{\rm{T}}})$, and ${{\bf{w}}_{\rm{R}}} \in {\cal W}({N_{\rm{R}}})$. It is noted that this signaling is based on the total transmission power, and we can further define the total transmission SNR as $\gamma_{\rm{tot}}=P_{\rm{tot}}/N_0$, and the received SNR with the total transmission power model as
\begin{equation}
\eta_{\rm{tot}}  = \gamma_{\rm{tot}}|{\bf{w}}_{\rm{R}}^{\rm{H}}{\bf{H}}{{\bf{w}}_{\rm{T}}}{|^2}.
\end{equation}

\begin{figure*}
\begin{equation} \label{eq_setW}
{\cal W}(N) = \{ \nu{[{\beta _1}{e^{{\rm{j}}{\theta _1}}},{\beta _2}{e^{{\rm{j}} {\theta _2}}},...,{\beta _N}{e^{{\rm{j}}{\theta _N}}}]^{\rm{T}}}\big{|}{\beta _i} \in \{ 0,1\} ,{\theta _i} \in [0,2\pi ),i = 1,2,...,N\}
\end{equation}
\hrulefill
\end{figure*}

The power gain under this model is
\begin{equation} \label{eq_power_gain_per}
G_{\rm{tot}}=\frac{\eta_{\rm{tot}}}{\gamma_{\rm{tot}}}=|{\bf{w}}_{\rm{R}}^{\rm{H}}{\bf{H}}{{\bf{w}}_{\rm{T}}}{|^2},
\end{equation}
which is also the array gain.

On the other hand, in mmWave communication the scaling abilities of PAs are generally limited. Thus, a per-antenna transmission power model is also with significance to characterize the best transmission ability of the transmitter, which is shown as
\begin{equation} \label{eq_signal_perp}
y = \sqrt {P_{\rm{per}}N_{\rm{Tact}}} {\bf{w}}_{\rm{R}}^{\rm{H}}{\bf{H}}{{\bf{w}}_{\rm{T}}}s + {\bf{w}}_{\rm{R}}^{\rm{H}}{\bf{n}},
\end{equation}
where $P_{\rm{per}}$ is the per-antenna transmission power, $N_{\rm{Tact}}$ is the number of active antennas of ${{\bf{w}}_{\rm{T}}}$, which varies as different ${{\bf{w}}_{\rm{T}}}$. Also, we have ${{\bf{w}}_{\rm{T}}} \in {\cal W}({N_{\rm{T}}})$, and ${{\bf{w}}_{\rm{R}}} \in {\cal W}({N_{\rm{R}}})$. In addition, the per-antenna transmission SNR is defined as $\gamma_{\rm{per}}=P_{\rm{per}}/N_0$, and the received SNR with the per-antenna transmission power model is defined as
\begin{equation}
\eta_{\rm{per}}  = \gamma_{\rm{per}}N_{\rm{Tact}}|{\bf{w}}_{\rm{R}}^{\rm{H}}{\bf{H}}{{\bf{w}}_{\rm{T}}}{|^2}.
\end{equation}

The power gain under this model is
\begin{equation}
G_{\rm{per}}=\frac{\eta_{\rm{per}}}{\gamma_{\rm{per}}}=N_{\rm{Tact}}|{\bf{w}}_{\rm{R}}^{\rm{H}}{\bf{H}}{{\bf{w}}_{\rm{T}}}{|^2},
\end{equation}
 which includes both the transmission power gain equal to the number of active antennas $N_{\rm{Tact}}$ and the array gain $|{\bf{w}}_{\rm{R}}^{\rm{H}}{\bf{H}}{{\bf{w}}_{\rm{T}}}{|^2}$.

It is worth mentioning that the total and per-antenna transmission power models are suitable for the cases that the scaling abilities of PA are high enough and limited, respectively. However, there is no difference for codebook design between with these two models.

\subsection{Channel Model}
Since mmWave channels are expected to have limited scattering \cite{rapp_2013_MMW,Rapp_2012_cellular_MMW,sayeed_2007,xiao2015Iterative,alkhateeb2014channel}, MPCs are mainly generated by reflection. That is, mmWave channels have the feature of directionality. Different MPCs have different physical transmit steering angles and receive steering angles, i.e., physical angles of departure (AoDs) and angles of arrival (AoAs). Consequently, mmWave channels are relevant to the geometry of antenna arrays. With half-spaced ULAs adopted at the transmitter and receiver, the channel matrix can be expressed as \cite{xiao2015Iterative,he2015suboptimal,alkhateeb2014channel,Ayach2014,hur2013millimeter,nsenga_2009}
\begin{equation} \label{eq_Channel}
{\bf{H}} = \sqrt {{N_{\rm{T}}}{N_{\rm{R}}}} \sum\limits_{\ell  = 1}^L {{\lambda _\ell }{\bf{a}}({N_{\rm{R}}},{\Omega _\ell }){\bf{a}}{{({N_{\rm{T}}},{\psi _\ell })}^{\rm{H}}}},
\end{equation}
where $\lambda_\ell$ is the complex coefficient of the $\ell$-th path, $L$ is the number of MPCs, ${\bf{a}}(\cdot)$ is the \emph{steering vector function}, ${\Omega _\ell }$ and ${\psi _\ell }$ are cos(AoD) and cos(AoA) of the $\ell$-th path, respectively. Let ${\theta _\ell }$ and ${\varphi _\ell }$ denote the \emph{physical AoD and AoA} of the $\ell$-th path, respectively; then we have ${\Omega _\ell } = \cos ({\theta _\ell })$ and ${\psi _\ell } = \cos ({\varphi _\ell })$. Therefore, ${\Omega _\ell }$ and ${\psi _\ell }$ are within the range $[-1~1]$. For convenience, in the rest of this paper, ${\Omega _\ell }$ and ${\psi _\ell }$ are called AoDs and AoAs, respectively. Similar to \cite{xiao2015Iterative,alkhateeb2014channel}, $\lambda_\ell$ can be modeled to be complex Gaussian distributed, while ${\theta _\ell }$ and ${\varphi _\ell }$ can be modeled to be uniformly distributed within $[0,2\pi]$. ${\bf{a}}(\cdot)$ is a function of the number of antennas and AoD/AoA, and can be expressed as
\begin{equation}
\begin{aligned}
{\bf{a}}(N,\Omega ) =\frac{1}{{\sqrt N }}[e^ {{\rm{j}}\pi 0\Omega},~e^{ {\rm{j}}\pi 1\Omega },...,e^{{\rm{j}}\pi (N - 1)\Omega}]^{\rm{T}},
\end{aligned}
\end{equation}
where $N$ is the number of antennas ($N$ is $N_{\rm{T}}$ at the transmitter and $M_{\rm{R}}$ at the receiver), $\Omega$ is AoD or AoA. It is easy to find that ${\bf{a}}(N,\Omega )$ is a periodical function which satisfies ${\bf{a}}(N,\Omega )={\bf{a}}(N,\Omega +2)$. The channel matrix ${\bf{H}}$ also has power normalization
\begin{equation}
\sum_{\ell=1}^{L}\mathbb{E}(|\lambda_\ell|^2)=1.
\end{equation}

\subsection{The Problem}
From a system level, joint Tx/Rx beamforming is required to maximize the received SNR, i.e.,
\begin{equation} \label{eq_beamforming}
\begin{aligned}
{\rm{Maximize}}~~~~&\eta_{\rm{tot}}  = \gamma_{\rm{tot}}|{\bf{w}}_{\rm{R}}^{\rm{H}}{\bf{H}}{{\bf{w}}_{\rm{T}}}{|^2}~\rm{or}\\
&\eta_{\rm{per}}  = \gamma_{\rm{per}}N_{\rm{Tact}}|{\bf{w}}_{\rm{R}}^{\rm{H}}{\bf{H}}{{\bf{w}}_{\rm{T}}}{|^2},\\
{\rm{Subject\;to}}~~~~&{{\bf{w}}_{\rm{R}}} \in {{\cal W}_{\rm{R}}},{{\bf{w}}_{\rm{T}}} \in {{\cal W}_{\rm{T}}}.
\end{aligned}
\end{equation}

Clearly, if ${\bf{H}}$ is known at the transmitter and receiver, and there is no CA constraint, the optimal ${\bf{w}}_{\rm{T}}$ and ${\bf{w}}_{\rm{R}}$ can be easily solved by singular value decomposition (SVD). However, in mmWave communication it is too time costly for entry-wise estimation of the channel matrix, which has a large scale due to large antenna arrays, and there exists the CA constraint. Thus, the SVD approach is basically not feasible for mmWave communication.

Fortunately, according to \eqref{eq_Channel} the mmWave channel has uncertainty mainly on AoDs/AoAs at both sides. In such a case, the one-stream beamforming problem in \eqref{eq_beamforming} can be simplified to find the AoD/AoA of an arbitrary strong MPC (or better the strongest MPC)\footnote{Basically, under LOS channel, where there is a LOS component significantly stronger than the other MPCs, the best MPC (the LOS component) needs to be found; while under NLOS channel, where all the MPCs have similar strengths, an arbitrary strong MPC can be feasible.}, and set ${\bf{w}}_{\rm{T}}$ and ${\bf{w}}_{\rm{R}}$ as the Tx/Rx steering vectors pointing to these AoD/AoA.

To this end, a straightforward way is to evenly sample the angle domain $[-1,1]$ with a small interval, e.g., $1/N$ for an $N$ antenna-array, and sequentially test all these sampled angles with corresponding steering vectors at both sides. This is the exhaustive search method. Clearly the codebooks for exhaustive search are composed by only steering vectors. Although exhaustive search is feasible and can always find the best MPC, it has a time complexity $\mathcal{O}(N^2)$ \cite{libin2013,he2015suboptimal}, which is too high for large arrays. Thus, hierarchical search is widely used to reduce the search time. In fact, the search time is highly dependent on the design of the Tx/Rx codebooks, which are subsets of ${{\cal W}_{\rm{T}}}$ and ${{\cal W}_{\rm{R}}}$, respectively. Hence, we focus on the codebook design in this paper.

\section{Hierarchical Codebook Design}
In this section, we design a hierarchical codebook composed by codewords (or AWVs) with different beam widths, which helps the search efficiency in finding the steering vectors of a strong or the strongest MPC at both sides. It is noted that, based on the specific structure of the mmWave channel model, the codebook design is to establish a relationship between the codewords in the angle domain to speed up the beam search. Thus, the codebook design is in fact irrelevant to the instantaneous channel response. When beamforming is required in practice before data transmission, a beam search process needs to be launched based on the designed codebook to find the suitable beamforming weights (steering vectors) for a given channel. For different channels, the codebook is the same, but the searched optimal steering vectors are different depending on the channel responses.

Although several hierarchial search schemes have been proposed for beam search in both literatures \cite{chen2011multi,hur2013millimeter,he2015suboptimal} and some standards, like IEEE 802.15.3c and IEEE 802.11ad \cite{wang_2009_beam_codebook,wang_2009_beam_codebook_vtc,Park_2010_11ad}, to the best of our acknowledge, there are no criteria proposed to judge whether a codebook is suitable or not, and there are few complete hierarchical codebooks with closed-form expressions provided for mmWave communication. Therefore, in this section, we first propose two basic criteria to design a hierarchical codebook, and then design one jointly using sub-array and deactivation techniques based on the proposed criteria.

\subsection{Two Criteria}
Before introducing the two criteria, we first introduce two definitions here. Let $A({\bf{w}},\Omega )$ denote the beam gain of ${\bf{w}}$ along angle $\Omega$, which is defined as
\begin{equation} \label{eq_beam_gain}
A({\bf{w}},\Omega ) = \sqrt N {\bf{a}}{(N,\Omega )^{\rm{H}}}{\bf{w}} = \sum\limits_{n = 1}^N {{{[{\bf{w}}]}_n}{e^{ - {\rm{j}}\pi (n - 1)\Omega }}},
\end{equation}
where $N$ is the number of elements of ${\bf{w}}$.

Let ${\cal C}{\cal V}({\bf{w}})$ denote the beam coverage in the angle domain of AWV ${\bf{w}}$, which can be mathematically expressed as
\begin{equation} \label{eq_beam_coverage}
{\cal C}{\cal V}({\bf{w}}) =\left \{ {\Omega}\big{|}\,|A({\bf{w}},{\Omega })| > \rho \mathop {\max }\limits_\omega  |A({\bf{w}},\omega )|\right\},
\end{equation}
where $\rho$ is a factor within $(0,1)$ to determine the beam coverage of ${\bf{w}}$. It is easy to find that the coverage is smaller as $\rho$ becomes greater. When $\rho=1/\sqrt{2}$, the beam coverage is the 3dB coverage, and the beam width is the well-known 3dB beam width. Different codebook design methods may have different values of $\rho$, and codewords with different beam widths in the same codebook may also have different values of $\rho$.

Hierarchical search is simply \emph{layered} search, i.e., the AWVs within the codebook are \emph{layered} according to their beam width. AWVs with a lower layer have larger beam width. Letting ${\bf{w}}(k,n)$ denote the $n$-th codeword (or AWV) in the $k$-th layer, the two criteria are presented as follows.

\textbf{Criterion 1:} The union of the beam coverage of all the codewords within each layer should cover the whole angle domain, i.e.,
\begin{equation}
\bigcup_{n = 1}^{{N_k}}{\cal C}{\cal V}({\bf{w}}(k,n)) = [ - 1,1],~k = 0,1,...,K,
\end{equation}
where $N_k$ is the number of codewords in the $k$-th layer, $K$ it the maximal index of the layer (there are $K+1$ layers in total).

\textbf{Criterion 2:} The beam coverage of an arbitrary codeword within a layer should be covered by the union of those of several codewords in the next layer, i.e.,
\begin{equation}
{\cal C}{\cal V}({\bf{w}}(k,n)) \subseteq  \bigcup _{m \in {\mathcal{I}_{k,n}}}{\cal C}{\cal V}({\bf{w}}(k + 1,m)), ~k = 0,1,...,K-1,
\end{equation}
where ${\mathcal{I}_{k,n}}$ is the index set with indices of the codewords in the $(k+1)$-th layer for the $n$-th codeword in the $k$-th layer. For convenience, we call ${\bf{w}}(k,n)$ a parent codeword, and $\{{\bf{w}}(k + 1,m)|m\in {\mathcal{I}_{k,n}}\}$ the child codewords of ${\bf{w}}(k,n)$.

It is clear that Criterion 1 guarantees the full coverage, i.e., there is no miss of any angle during the beam search, while Criterion 2 establishes a tree-fashion relationship between the codewords, which enables hierarchical search. If each parent codeword has $M$ child codewords, all the codewords in the codebook constitute an $M$-way tree with respect to their beam coverage in the angle domain. In such a case, hierarchical search can be easily realized by using the tree search algorithm in both the receiver and transmitter as following \cite{alkhateeb2014channel,he2015suboptimal}.

\emph{The Hierarchical Search:} Initially, we fix the transmitter to be in an omni-directional mode, and run an $M$-way tree search for $\log_M(N_{\rm{R}})$ stages to find the best receive codeword. And then we fix the receiver to be in a directional mode with the found best receive codeword, and then run an $M$-way tree search for $\log_M(N_{\rm{T}})$ stages to find the best transmit codeword. In each stage, we have $M$ candidate codewords, which are the $M$ child codewords of a parent codeword found in the last stage. We need to test all the $M$ codewords one by one to find the best one, and treat it as a new parent codeword for the next-stage search. Therefore, the search time (number of tests) for Tx/Rx joint training is
\begin{equation}
T=M\log_M{N_{\rm{T}}}+M\log_M{N_{\rm{R}}}.
\end{equation}

 In the next subsection we will design a codebook with $M=2$, for the reason that when $M=2$ the codebook tree is a typical binary tree, and the number of antennas is powers of two, which is generally used in antenna array design. Nevertheless, extending the proposed method to other values of $M$ is straightforward.

\subsection{The Deactivation Approach}
 As a basis of the joint sub-array and deactivation approach, we first introduce the deactivation (DEACT) approach in this subsection to design a binary-tree codebook, which has the beam coverage shown in Fig. \ref{fig:codebook}, where there are $\log_2(N)+1$ layers with indices from $k=0$ to $k=\log_2(N)$, and the number of codewords in the $k$-th layer $N_k=2^k$. Here $N$ denotes the number of antennas of an arbitrary array. Thus, $N=N_{\rm{T}}$ at the transmitter and $N=N_{\rm{R}}$ at the receiver. Besides, we have
\begin{equation}
\begin{aligned}
&{\cal C}{\cal V}({\bf{w}}(k,n))= {\cal C}{\cal V}({\bf{w}}(k + 1,2n - 1)) \cup {\cal C}{\cal V}({\bf{w}}(k + 1,2n)),\\
&~~~~~k = 0,1,...,(\log_2(N)-1),~n = 1,2,3,...,{2^k}.
\end{aligned}
\end{equation}

\begin{figure}[t]
\begin{center}
  \includegraphics[width=\figwidth cm]{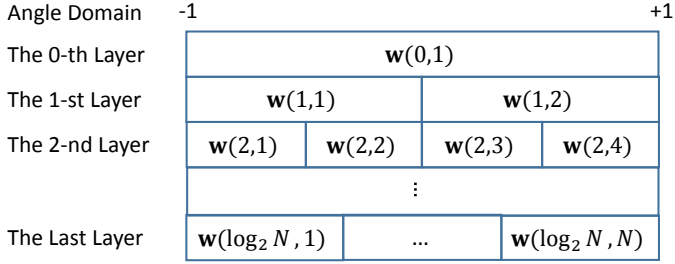}
  \caption{Beam coverage of a binary-tree structured codebook.}
  \label{fig:codebook}
\end{center}
\end{figure}

In our method, we define
\begin{equation} \label{eq_beamcover_steer}
{\cal C}{\cal V}({\bf{a}}(N,\Omega )) = \left[\Omega  - \frac{1}{N},\Omega  + \frac{1}{N}  \right],
\end{equation}
which means that the steering vectors have a beam width $2/N$ centering at the steering angle \cite{TseFundaWC}. In other words, within the beam coverage of ${\bf{a}}(N,\Omega )$, it has the maximal beam gain along the angle $\Omega$, while the minimal beam gain along the angles $\Omega\pm 1/N$. Thus, we can compute the value of $\rho$ for our codebook as
\begin{equation}
\begin{aligned}
\rho&=\left|\frac{{\bf{a}}(N,\Omega-1/N)^{\rm{H}}{\bf{a}}(N,\Omega)}{{\bf{a}}(N,\Omega)^{\rm{H}}{\bf{a}}(N,\Omega)}\right|\\
&~~~~{\rm{or}}~\left|\frac{{\bf{a}}(N,\Omega+1/N)^{\rm{H}}{\bf{a}}(N,\Omega)}{{\bf{a}}(N,\Omega)^{\rm{H}}{\bf{a}}(N,\Omega)}\right|\\
&=\frac{1}{N}\Big|\sum_{n=1}^{N}e^{{\rm{j}}\pi(n-1)/N}\Big|.
\end{aligned}
\end{equation}
Although the value of $\rho$ depends on $N$, we have $\rho\approx 0.64$ given that $N$ is large, e.g., $N\geq 8$. Even when $N$ is small, $\rho$ is still close to 0.64, e.g., $\rho=0.65$ when $N=4$.

Notice that the $N$ codewords in the last layer cover an angle range $[-1,1]$ in total, which means that all these codewords must have the narrowest beam width $2/N$ with different steering angles. In other words, the codewords in the last layer should be the steering vectors with angles evenly sampled within $[-1,1]$. Consequently, we have ${\cal C}{\cal V}({\bf{w}}(\log_2(N),n))=[ - 1 + \frac{{2n - 2}}{{{N}}}, - 1 + \frac{{2n}}{{{N}}}]$, $n=1,2,...,N$. With the beam coverage of the last-layer codewords, we can further obtain that of the codewords in the other layers in turn as an order of descending layer indices, i.e., obtain ${\cal C}{\cal V}({\bf{w}}(\log_2(N)-1,n)),~{\cal C}{\cal V}({\bf{w}}(\log_2(N)-2,n)),~...,~{\cal C}{\cal V}({\bf{w}}(0,n))$ in turn. Finally, the beam coverage of all the codewords can be uniformly written as
\begin{equation} \label{eq_beamcover_w}
\begin{aligned}
{\cal C}{\cal V}({\bf{w}}(k,n)) = [ - 1 + \frac{{2n - 2}}{{{2^k}}}, - 1 + \frac{{2n}}{{{2^k}}}],\\
k=1,2,...,\log_2 N,~n = 1,2,3,...,{2^k}.
\end{aligned}
\end{equation}

Comparing \eqref{eq_beamcover_w} with \eqref{eq_beamcover_steer}, it is clear that when
\begin{equation} \label{eq_codeword_DEACT}
{\bf{w}}(k,n)=[{\bf{a}}(2^k,-1+\frac{2n-1}{2^k})^{\rm{T}},{\bf{0}}_{(N-2^k)\times 1}^{\rm{T}}]^{\rm{T}},
\end{equation}
\eqref{eq_beamcover_w} is satisfied. This is just the deactivation approach that was proposed in \cite{he2015suboptimal}, where the number of active antennas is $2^k$ in the $k$-th layer, and the other antennas are all turned off. Fig. \ref{fig:beam_Deact} shows an example of beam pattern of the DEACT approach for the case of $N=128$. From this figure we find that the beam coverage of ${\bf{w}}(0,1)$ is just the union of those of ${\bf{w}}(1,1)$ and ${\bf{w}}(1,2)$, while the beam coverage of ${\bf{w}}(1,1)$ is just the union of those of ${\bf{w}}(2,1)$ and ${\bf{w}}(2,2)$.

\begin{figure}[t]
\begin{center}
  \includegraphics[width=\figwidth cm]{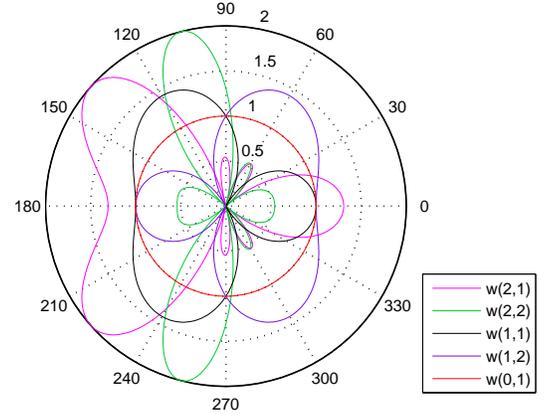}
  \caption{Beam patterns of ${\bf{w}}(2,1)$, ${\bf{w}}(2,2)$, ${\bf{w}}(1,1)$, ${\bf{w}}(1,2)$ and ${\bf{w}}(0,1)$ for the DEACT approach, where $N=128$.}
  \label{fig:beam_Deact}
\end{center}
\end{figure}

\subsection{The Joint Sub-Array and Deactivation Approach}

It is noted that for the deactivation approach, when $k$ is small, the number of active antennas is small, or even 1 when $k=0$. This greatly limits the maximal total transmission power of an mmWave device. In general, we hope the number of active antennas is as large as possible, such that higher power can be transmitted, because in mmWave communication per-antenna transmission power is usually limited. To achieve this target, we consider jointly using the sub-array and deactivation approach here. As the key of this approach is beam widening via single-RF subarray, we term it BMW-SS.

We also want to design a codebook with the beam coverage shown in Fig. \ref{fig:codebook}. According to \eqref{eq_beamcover_steer}, in the $k$-th layer, each codeword has a beam width of $2/2^k$. For the codewords of the last layer, we can also adopt the steering vectors according to \eqref{eq_codeword_DEACT}. Compared with the codewords in the last layer, those in the lower layers have wider beams, and according to \eqref{eq_beamcover_w}, codewords in the same layer have the same beam widths but different steering angles. Thus, there are two basic tasks in the codebook design, namely to rotate the beam along required directions and to broaden the beam by required factors. We first introduce beam rotation.

\subsubsection{Beam Rotation} Beam rotation can be realized according to the following theorem.

\begin{theorem}
${\cal C}{\cal V}({\bf{w}} \circ \sqrt N {\bf{a}}(N,\psi )) = {\cal C}{\cal V}({\bf{w}}) + \psi$, where $\circ$ represents entry-wise product (a.k.a. Hadamard product), $N$ is the number of elements of ${\bf{w}}$, $\psi$ is an arbitrary angle. $\mathcal{A}+\psi$ is a new angle set with elements being those of the angle set $\mathcal{A}$ added by $\psi$.
\end{theorem}

The proof is referred to Appendix A, and this theorem can be used not only for the BMW-SS approach, but also for other codebook design methods.

Theorem 1 implies that given an arbitrary codeword ${\bf{w}}$, we can rotate its beam coverage ${\cal C}{\cal V}({\bf{w}})$ by $\psi$ with ${\bf{w}} \circ \sqrt N {\bf{a}}(N,\psi )$. This theorem helps to design all the other codewords in the same layer once one codeword in this layer is found. To explain this, we need to emphasize that all the codewords in the same layer have the same beam widths but different steering angles according to \eqref{eq_beamcover_w}, which means that the beam coverage of all the codewords can be assumed to have the same shape but different offsets in the angle domain. Thus, we can obtain one codeword based on another in the same layer as long as we know the angle gap between them according to Theorem 1. In particular, suppose we find the first codeword in the $k$-th layer ${\bf{w}}(k,1)$. According to \eqref{eq_beamcover_w}, we do know that the angle gap between the $n$-th codeword in the $k$-th layer, i.e., ${\bf{w}}(k,n)$, and ${\bf{w}}(k,1)$ is $\frac{2n-2}{2^k}$, $n=2,3,...,2^k$. Then we can obtain the all the other codewords in this layer based on ${\bf{w}}(k,1)$ according to Theorem 1 (see Corollary 1 below).

\begin{corollary}
Given the first codeword in the $k$-th layer ${\bf{w}}(k,1)$, all the other codewords in the $k$-th layer can be found through rotating ${\bf{w}}(k,1)$ by $\frac{2n-2}{2^k}$, $n=2,3,...,2^k$, respectively, i.e., ${\bf{w}}(k,n)={\bf{w}}(k,1) \circ \sqrt N {\bf{a}}(N, \frac{2n-2}{2^k})$.
\end{corollary}
\begin{proof}
To prove this corollary, we need to prove that, according to \eqref{eq_beamcover_w}, when ${\bf{w}}(k,n)={\bf{w}}(k,1) \circ \sqrt N {\bf{a}}(N, \frac{2n-2}{2^k})$, ${\bf{w}}(k,n)\in\mathcal{W}(N)$ and ${\cal C}{\cal V}({\bf{w}}(k,n))=[-1+\frac{2n-2}{2^k},-1+\frac{2n}{2^k}]$.

Since
\begin{equation}
\begin{aligned}
&[{\bf{w}}(k,n)]_i=[{\bf{w}}(k,1) \circ \sqrt N {\bf{a}}(N,\frac{2n-2}{2^k} )]_i\\
=&[{\bf{w}}(k,1)]_ie^{{\rm{j}}\pi(n-1)\frac{2n-2}{2^k}},
\end{aligned}
\end{equation}
we have $|[{\bf{w}}(k,n)]_i|=|[{\bf{w}}(k,1)]_i|$. As ${\bf{w}}(k,1)\in\mathcal{W}(N)$, ${\bf{w}}(k,n)\in\mathcal{W}(N)$.

In addition, ${\bf{w}}(k,1)$ has a beam coverage $[-1,-1+\frac{2}{2^k}]$. According to Theorem 1,
\begin{equation}
\begin{aligned}
&{\cal C}{\cal V}({\bf{w}}(k,n))\\
=&{\cal C}{\cal V}({\bf{w}}(k,1) \circ \sqrt N {\bf{a}}(N,\frac{2n-2}{2^k} ))\\
=&{\cal C}{\cal V}({\bf{w}}(k,1))+\frac{2n-2}{2^k}\\
=&[-1,-1+\frac{2}{2^k}]+\frac{2n-2}{2^k}\\
=&[-1+\frac{2n-2}{2^k},-1+\frac{2n}{2^k}].
\end{aligned}
\end{equation}
\end{proof}

\subsubsection{Beam Broadening}
The remaining task is to broaden the beam for each layer. Given an $N$-element array, generally we would expect a beam width of $2/N$. Nevertheless, this beam width is in fact achieved by concentrating the transmission power at a specific angle $\Omega_0$, i.e., by selecting AWV to maximize $|A({\bf{w}},\Omega_0)|$. Intuitively, if we design the AWV to disperse the transmission power along different widely-spaced angles, the beam width can be broadened. More specifically, if a large antenna array is divided into multiple sub-arrays, and these sub-arrays point at sufficiently-spaced directions, a wider beam can be shaped.

To illustrate this, let us separate the $N$-antenna array into $M$ sub-arrays with $N_{\rm{S}}$ antennas in each sub-array, which means $N=MN_{\rm{S}}$. In addition, letting ${{\bf{f}}_m} = {[{\bf{w}}]_{(m - 1){N_{\rm{S}}} + 1:m{N_{\rm{S}}}}}$, we have ${[{{\bf{f}}_m}]_n} = {[{\bf{w}}]_{(m - 1){N_{\rm{S}}} + n}}$, $m=1,2,...,M$. ${{\bf{f}}_m}$ can be seen as the sub-AWV of the $m$-th sub-array. Therefore, the beam gain of ${\bf{w}}$ writes
\begin{equation}
\begin{aligned}
A({\bf{w}},\omega )& = \sum\limits_{n = 1}^N {{{[{\bf{w}}]}_n}{e^{ - {\rm{j}}\pi (n - 1)\omega }}} \\
 &= \sum\limits_{m = 1}^M {\sum\limits_{n = 1}^{{N_{\rm{S}}}} {{{[{\bf{w}}]}_{(m - 1){N_{\rm{S}}} + n}}{e^{ - {\rm{j}}\pi ((m - 1){N_{\rm{S}}} + n - 1)\omega }}} } \\
 &= \sum\limits_{m = 1}^M {\sum\limits_{n = 1}^{{N_{\rm{S}}}} {{e^{ - {\rm{j}}\pi (m - 1){N_{\rm{S}}}\omega }}{{[{{\bf{f}}_m}]}_n}{e^{ - {\rm{j}}\pi (n - 1)\omega }}} }\\
& = \sum\limits_{m = 1}^M {{e^{ - {\rm{j}}\pi (m - 1){N_{\rm{S}}}\omega }}A({{\bf{f}}_m},\omega )},
\end{aligned}
\end{equation}
which means that the beam gain of ${\bf{w}}$ can be seen as the union of those of ${{\bf{f}}_m}$. According to \eqref{eq_beamcover_steer}, by assigning ${{\bf{f}}_m}=e^{{\rm{j}}\theta_m}{\bf{a}}({N_{\rm{S}}},-1+\frac{2m-1}{{N_{\rm{S}}}})$, where $e^{{\rm{j}}\theta_m}$ can be seen as a scalar coefficient with unit norm for the $m$-th sub-array, the $m$-th sub-array has beam coverage $\mathcal{CV}({{\bf{f}}_m})=[-1+\frac{2m-2}{{N_{\rm{S}}}},-1+\frac{2m}{{N_{\rm{S}}}}]$, $m=1,2,...,M$. Hence, ${\bf{w}}$ has the beam coverage
\begin{equation}
\mathcal{CV}({{\bf{w}}})=\bigcup_{m=1}^{M}\mathcal{CV}({{\bf{f}}_m})=[-1,-1+\frac{2M}{{N_{\rm{S}}}}]=[-1,-1+\frac{2M^2}{N}],
\end{equation} i.e., the beam width has been broadened by $M^2$ by using the sub-array technique, where a broadening factor $M$ comes from the number of sub-arrays, while another factor $M$ results from the reduction factor of the sub-array size.

However, in the above process, the mutual effects between different sub-arrays are not taken into account. In the case of ${{\bf{f}}_m}=e^{{\rm{j}}\theta_m}{\bf{a}}({N_{\rm{S}}},-1+\frac{2m-1}{{N_{\rm{S}}}})$, we have
\begin{equation}
\begin{aligned}
&A({\bf{w}},\omega )~\big{|}~{{\bf{f}}_m} = e^{{\rm{j}}\theta_m}{\bf{a}}({N_{\rm{S}}}, - 1 + \frac{{2m - 1}}{{{N_{\rm{S}}}}})\\
 =&\sqrt {{N_{\rm{S}}}}\sum\limits_{m = 1}^M {e^{ - {\rm{j}}\pi (m - 1){N_{\rm{S}}}\omega }}e^{{\rm{j}}\theta_m}\times \\
 &\quad\quad\quad {\bf{a}}{{({N_{\rm{S}}},\omega )}^{\rm{H}}}{\bf{a}}({N_{\rm{S}}}, - 1 + \frac{{2m - 1}}{{{N_{\rm{S}}}}}).
\end{aligned}
\end{equation}

As the steering vector has the properties that ${\bf{a}}({N_{\rm{S}}}, - 1 + \frac{{2m - 1}}{{{N_{\rm{S}}}}})^{\rm{H}}{\bf{a}}({N_{\rm{S}}}, - 1 + \frac{{2n - 1}}{{{N_{\rm{S}}}}})=0$ when $m\neq n$, the beam gain of ${{\bf{f}}_m}$ along the angle $- 1 + \frac{{2m - 1}}{{{N_{\rm{S}}}}}$ is affected little by ${{\bf{f}}_n}$. It is clear that $|A({\bf{w}}, - 1 + \frac{{2m - 1}}{{{N_{\rm{S}}}}})|=\sqrt{{N_{\rm{S}}}}$ for $m=1,2,...,M$, which means that the beam gains along angles $\omega= 1 + \frac{{2m - 1}}{{{N_{\rm{S}}}}}$ are significant.

Additionally, to reduce the beam fluctuation, it is required that the intersection points in the angle domain of these coverage regions, i.e., $\omega=- 1 + \frac{{2n}}{{{N_{\rm{S}}}}}$, $n=1,2,...,M-1$, also have high beam gain, and this can be realized by adjusting coefficients $e^{{\rm{j}}\theta_m}$. Concretely, we have \eqref{eq_subarray_dev}, where in (a) we have used the fact that ${\bf{a}}{({N_{\rm{S}}},{\omega _1})^{\rm{H}}}{\bf{a}}({N_{\rm{S}}},{\omega _2})$ is small and can be neglected when $|\omega_1-\omega_2|>2/{N_{\rm{S}}}$, in (b) we have exploited the condition that ${N_{\rm{S}}}$ is even in this paper. To maximize $|A({\bf{w}},\omega )|^2$, we face the problem
\begin{equation}
\mathop {{\rm{maximize}}}\limits_{\Delta \theta } \;~~|f({N_{\rm{S}}},\Delta \theta ){|^2},
\end{equation}
which has a solution that $\Delta\theta=(2k-\frac{{N_{\rm{S}}}-1}{{N_{\rm{S}}}})\pi$, where $k\in\mathbb{Z}$. Thus, we may choose $\theta_m=-{jm\frac{{N_{\rm{S}}}-1}{{N_{\rm{S}}}}\pi}$, which satisfies $\Delta\theta=\pi$, to reduce the fluctuation of the beam.

\begin{figure*}
\begin{equation} \label{eq_subarray_dev}
\begin{aligned}
&A({\bf{w}},\omega )~\big{|}~\big\{{{\bf{f}}_m} = {e^{{\rm{j}}{\theta _m}}}{\bf{a}}({N_{\rm{S}}}, - 1 + \frac{{2m - 1}}{{{N_{\rm{S}}}}}),~\omega  =  - 1 + \frac{{2n}}{{{N_{\rm{S}}}}}\big\} = \sqrt {{N_{\rm{S}}}} \sum\limits_{m = 1}^M {{e^{ - {\rm{j}}\pi (m - 1){N_{\rm{S}}}\omega }}{e^{{\rm{j}}{\theta _m}}}{\bf{a}}{{({N_{\rm{S}}},\omega )}^{\rm{H}}}{\bf{a}}({N_{\rm{S}}}, - 1 + \frac{{2m - 1}}{{{N_{\rm{S}}}}})} \\
\mathop  \approx \limits^{(a)}& \sqrt {{N_{\rm{S}}}} {e^{ - {\rm{j}}\pi (n - 1){N_{\rm{S}}}\omega }}{e^{{\rm{j}}{\theta _n}}}{\bf{a}}{({N_{\rm{S}}}, - 1 + \frac{{2n}}{{{N_{\rm{S}}}}})^{\rm{H}}}{\bf{a}}({N_{\rm{S}}}, - 1 + \frac{{2n - 1}}{{{N_{\rm{S}}}}}) +\sqrt {{N_{\rm{S}}}} {e^{ - {\rm{j}}\pi n{N_{\rm{S}}}\omega }}{e^{{\rm{j}}{\theta _{n + 1}}}}{\bf{a}}{({N_{\rm{S}}}, - 1 + \frac{{2n}}{{{N_{\rm{S}}}}})^{\rm{H}}}{\bf{a}}({N_{\rm{S}}}, - 1 + \frac{{2n + 1}}{{{N_{\rm{S}}}}})\\
 =& \frac{1}{{\sqrt {{N_{\rm{S}}}} }}{e^{ - {\rm{j}}\pi (n - 1){N_{\rm{S}}}\omega }}{e^{{\rm{j}}{\theta _n}}}\times\left( {\sum\limits_{i = 1}^{{N_{\rm{S}}}} {{e^{ - {\rm{j}}\pi (i - 1)/{N_{\rm{S}}}}}}  + {e^{{\rm{j}}\pi {N_{\rm{S}}}\omega }}{e^{{\rm{j}}({\theta _{n + 1}} - {\theta _n})}}\sum\limits_{i = 1}^{{N_{\rm{S}}}} {{e^{{\rm{j}}\pi (i - 1)/{N_{\rm{S}}}}}} } \right)\\
\mathop  = \limits^{(b)}& \frac{1}{{\sqrt {{N_{\rm{S}}}} }}{e^{ - {\rm{j}}\pi (n - 1){N_{\rm{S}}}\omega }}{e^{{\rm{j}}{\theta _n}}}\times\left( {\sum\limits_{i = 1}^{{N_{\rm{S}}}} {{e^{ - {\rm{j}}\pi (i - 1)/{N_{\rm{S}}}}}}  + {e^{{\rm{j}}\Delta \theta }}\sum\limits_{i = 1}^{{N_{\rm{S}}}} {{e^{{\rm{j}}\pi (i - 1)/{N_{\rm{S}}}}}} } \right)\\
\buildrel \Delta \over =& \frac{1}{{\sqrt {{N_{\rm{S}}}} }}{e^{ - {\rm{j}}\pi (n - 1){N_{\rm{S}}}\omega }}{e^{{\rm{j}}{\theta _n}}}f({N_{\rm{S}}},\Delta \theta ),
\end{aligned}
\end{equation}
\hrulefill
\end{figure*}

In summary, by using the sub-array method and setting ${{\bf{f}}_m}=e^{-jm\frac{{N_{\rm{S}}}-1}{{N_{\rm{S}}}}\pi}{\bf{a}}({N_{\rm{S}}},-1+\frac{2m-1}{{N_{\rm{S}}}})$, $m=1,2,...,M$, we obtain a codeword ${\bf{w}}$ with a beam width $\frac{2M}{N_{\rm{S}}}=\frac{2M^2}{N}$. If we jointly using the sub-array and deactivation method, we may obtain codewords with beam widths $\frac{2N_{\rm{A}}}{N_{\rm{S}}}=\frac{2MN_{\rm{A}}}{N}$ by setting as
\begin{equation} \label{eq_codeword_setting}
{{\bf{f}}_m}=\left\{
\begin{aligned}
&e^{-jm\frac{{N_{\rm{S}}}-1}{{N_{\rm{S}}}}\pi}{\bf{a}}({N_{\rm{S}}},-1+\frac{2m-1}{{N_{\rm{S}}}}), ~m=1,2,...,N_{\rm{A}},\\
&{\bf{0}}_{{N_{\rm{S}}}\times 1}, ~~m=N_{\rm{A}}+1,N_{\rm{A}}+2,...,M.
\end{aligned} \right.
\end{equation}
where $N_{\rm{A}}$ is the number of active sub-arrays.

\subsubsection{Codebook Generation}

Recall that we need to design ${\bf{w}}(k,n)$ with beam widths $2/2^k$ in the $k$-th layer.

When $k=\log_2(N)$, we have ${\bf{w}}(\log_2(N),n)={\bf{a}}(N,-1+\frac{2n-1}{N})$, $n=1,2,...,N$.

When $k=\log_2(N)-\ell$, where $\ell=1,2,...,\log_2(N)$, we obey the following procedures to compute ${\bf{w}}(k,n)$:

\begin{itemize}
  \item {Separate ${\bf{w}}(k,1)$ into $M=2^{\lfloor(\ell+1)/2\rfloor}$ sub-arrays with ${{\bf{f}}_m} = {[{\bf{w}}(k,1)]_{(m - 1){N_{\rm{S}}} + 1:m{N_{\rm{S}}}}}$, where $\lfloor\cdot\rfloor$ is the flooring integer operation, $m=1,2,...,M$;}
  \item {Set ${\bf{f}}_m$ as \eqref{eq_codeword_setting}, where $N_{\rm{A}}=M/2$ if $\ell$ is odd, and $N_{\rm{A}}=M$ if $\ell$ is even;}
  \item {According to Corollary 1, we have ${\bf{w}}(k,n)={\bf{w}}(k,1)\circ \sqrt N {\bf{a}}(N,\frac{2(n-1)}{N})$, where $n=2,3,...,2^k$;}
  \item {Normalize ${\bf{w}}(k,n)$.}
\end{itemize}

Fig. \ref{fig:beam_Jot} shows an example of the beam pattern of the BMW-SS approach in the case of $N=128$. From this figure we find that the beam coverage of ${\bf{w}}(0,1)$ is just the union of those of ${\bf{w}}(1,1)$ and ${\bf{w}}(1,2)$, while the beam coverage of ${\bf{w}}(1,1)$ is just the union of those of ${\bf{w}}(2,1)$ and ${\bf{w}}(2,2)$. Comparing the beam pattern of DEACT shown in Fig. \ref{fig:beam_Deact} with that in Fig. \ref{fig:beam_Jot}, it can be observed that although there are small-scale fluctuations for BMW-SS, the beams of BMW-SS appear flatter than those of DEACT within the covered angle.

 On the other hand, for BMW-SS all the codewords either have all the antennas activated, or have half of them activated, which shows a significant advantage over DEACT in terms of the maximal total transmission power, especially for the low-layer codewords. Fig. \ref{fig:beam_Cmp} shows the comparison of beam patterns of BMW-SS, DEACT, and the approach in \cite{alkhateeb2014channel} (termed as Sparse) with the per-antenna transmission power model, where all the weights of active antennas have a unit amplitude. From this figure, we find that BMW-SS offers much higher beam gains than DEACT due to exploiting much greater number of active antennas. In addition, for the Sparse codebook, when the number of RF chains is small, there are deep sinks within the beam coverage of the wide-beam codewords, and the sink is more severe when the number of RF chains is smaller, which are in accordance with the results in \cite{alkhateeb2014channel} (Fig. 5 therein). Clearly, if the AoD or AoA of an MPC is along the sink angle, it cannot be detected with the codeword, which results in miss detection of the MPC. By contrast, BMW-SS does not have such deep sinks.

\begin{figure}[t]
\begin{center}
  \includegraphics[width=\figwidth cm]{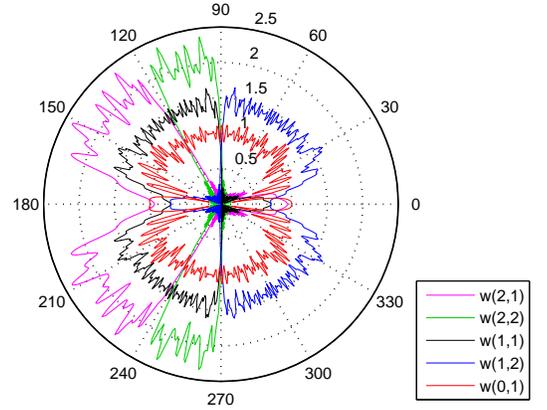}
  \caption{Beam patterns of ${\bf{w}}(2,1)$, ${\bf{w}}(2,2)$, ${\bf{w}}(1,1)$, ${\bf{w}}(1,2)$ and ${\bf{w}}(0,1)$ for the BMW-SS approach, where $N=128$.}
  \label{fig:beam_Jot}
\end{center}
\end{figure}

\begin{figure}[t]
\begin{center}
  \includegraphics[width=\figwidth cm]{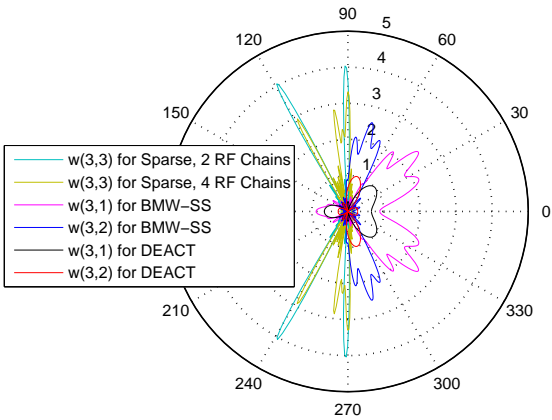}
  \caption{Comparison of the beam patterns of BMW-SS, DEACT, and the approach in \cite{alkhateeb2014channel} (termed as Sparse) with the per-antenna transmission power model, where $N=32$. $L_d=1$ for the Sparse approach.}
  \label{fig:beam_Cmp}
\end{center}
\end{figure}

It is noted that the corresponding hierarchical search of the designed codebook will eventually converge to a codeword of the last layer, i.e., a steering vector, at both ends. We can find that the angle resolution of the last layer is $2/N$. Thus, the designed codebook is just coarse codebook, while the corresponding search method is coarse search, like those in \cite{he2015suboptimal}. If a higher angle resolution is required, a fine codebook composed by steering vectors with a smaller sampling gap than $2/N$ is necessary. Details are referred to \cite{he2015suboptimal}.

\subsection{Generalization}

\subsubsection{For the Hybrid Precoding Structure}
In this paper we adopt an analog beamforming structure, and both the proposed two criteria and the BMW-SS approach are based on the analog beamforming structure. However, they are naturally feasible for the hybrid precoding structure, because the analog beamforming structure can be seen as one of the branches of the hybrid precoding structure \cite{alkhateeb2014channel}.

To realize multi-stream transmission with the hybrid precoding structure, the AoDs and AoAs of multiple MPCs need to be searched. The search process with BMW-SS based on the beamforming structure in this paper can be adopted to search the AoD and AoA of each single MPC. In fact, similar extension from one-stream transmission to multi-stream transmission has been studied for the Sparse codebook in \cite{alkhateeb2014channel}.

\subsubsection{For Other Types of Antenna Arrays}
In this paper we adopt a ULA model. There are other types of antenna arrays in practice, e.g., uniform planar array (UPA) and uniform circular array (UCA). The proposed criteria and the BMW-SS approach can be easily extended to the UPA model. In particular, for a typical 2-dimensional grid UPA with $m\times n$ elements, its steering vector can be expressed as the Kronecker product of those of two ULAs with $m\times 1$ and $n\times 1$ elements, respectively \cite{hur2013millimeter}. The search process as well as the codebook design could be extended to the UPA model and will be studied later.

On the other hand, for a UCA model, the two criteria are also feasible, but that the beam coverage in Criterion 1 should be extended to a 2-dimensional angle range, including both the azimuth and elevation angle ranges. However, the proposed BMW-SS approach can hardly be extended to the UCA model, because the relation between the elements of a steering vector changes. It would be indeed interesting to design a new codebook according to the proposed criteria with a UCA mode.

\subsubsection{For Arbitrary Number of Antenna Elements}
In this paper we require that the number of elements of a ULA ($N$) is $M^p$ for some positive integer $p$, which is because the BMW-SS approach needs to divide the array or a sub-array into $M$ smaller sub-arrays. For a ULA with an arbitrary number of elements, the sub-array technology is infeasible if $N$ is not $M$ to an integer power. Hence, the proposed approach may not be extended to with arbitrary number of antenna elements. There are two possible choices in practice. One is to select a ULA with $N$ being $M$ to an integer power when designing the system, which is reasonable because the beamforming method should be considered in system planning. The other one is to exploit BMW-SS for beamforming with $M^{\lfloor\log_M N\rfloor}$ antennas while deactivating the other ones, where $\lfloor\cdot \rfloor$ is the floor operation. Afterwards, further beam refinement can be launched with all the antennas activated.

\section{Performance Evaluation}
In this section, we evaluate the performance of the designed hierarchical codebook by the BMW-SS approach, and compare it with the alternatives. We consider two different system models on the transmission power, namely total transmission power and per-antenna transmission power, which correspond to the signal models in \eqref{eq_signal_totalp} and \eqref{eq_signal_perp}, respectively. The total transmission power signal model reflects the performance for the case that the transmission power on each antenna branch can be high enough, while the per-antenna transmission power signal model reflects the limit performance for the case that the transmission power on each antenna branch is limited\footnote{The limit performance here refers to the performance in the case that all the active antennas transmit with maximal power.}. The per-antenna transmission power model makes more sense in mmWave communication, where the output power of a single power amplifier is generally limited \cite{wang_2011_MMWCS,Xia_2011_60GHz_Tech}. The activation/deactivation operations of a codebook are irrelevant to the power models. In particular, no matter which power model is adopted, the codewords of BMW-SS either have all or half of the antennas activated, those of DEACT have varying numbers of the antennas activated and the number may be quite small, while those of Sparse always have all the antennas activated.

Besides, in the simulations, both LOS and NLOS channel models are considered based on \eqref{eq_Channel}. For LOS channel, the first MPC has a constant coefficient and random AoD and AoA, while the other NLOS MPCs have complex Gaussian-distributed coefficients and random AoDs and AoAs \cite{hur2013millimeter,xiao2015Iterative}. The LOS MPC is generally much stronger than the NLOS MPCs. For NLOS channel, all the MPCs have complex Gaussian-distributed coefficients with the same variance and random AoDs and AoAs \cite{hur2013millimeter,xiao2015Iterative,alkhateeb2014channel}. Both the LOS and NLOS channels are sparse in the angle domain, because the number of MPCs is much smaller than the numbers of the Tx/Rx antennas \cite{hur2013millimeter,xiao2015Iterative,alkhateeb2014channel}. For all the codebooks, the hierarchical search method introduced in Section III-A is used. The performances of received SNR and success rate are all averaged on the instantaneous results of $10^4$ random realizations of the LOS/NLOS channel.

\subsection{Total Transmission Power Model}
In this subsection, the total transmission power signal model shown in \eqref{eq_signal_totalp} is used. With this model, the deactivation of antennas will not affect the total transmission power, i.e., the total transmission power is the same for the involved schemes.

Fig. \ref{fig:received_power} shows the received power during each search step with the BMW-SS and DEACT codebooks under both LOS and NLOS channels, where $N_{\rm{T}}=N_{\rm{R}}=64$, $L=3$, $P_{\rm{tot}}=1$ W, and $N_0=10^{-4}$ W, i.e., the SNR for beam training is sufficiently high, which means the length of the training sequence is sufficiently long. The upper bound is achieved by the exhaustive search method. For the LOS channel, the LOS component has 15dB higher power than that of an NLOS MPC. From this figure we can find that the received-power performance of these two codebooks is similar to each other. Under both channels, at the beginning, i.e., in the first two steps, DEACT behaves slightly better than BMW-SS; while in the following steps, BMW-SS slightly outperforms DEACT, until both methods achieve the same performance after the search process, because they have the same last-layer codewords. Meanwhile, both approaches reach the upper bound under LOS channel, while achieve a performance close to the upper bound under NLOS channel. This is because under LOS channel, the LOS component is the optimal MPC, and it is acquired by all BMW-SS, DEACT and the exhaustive search. However, under NLOS channel, BMW-SS and DEACT acquire an arbitrary MPC of the $L$ NLOS MPCs, which may not be the optimal one acquired by the exhaustive search.

\begin{figure}[t]
\begin{center}
  \includegraphics[width=\figwidth cm]{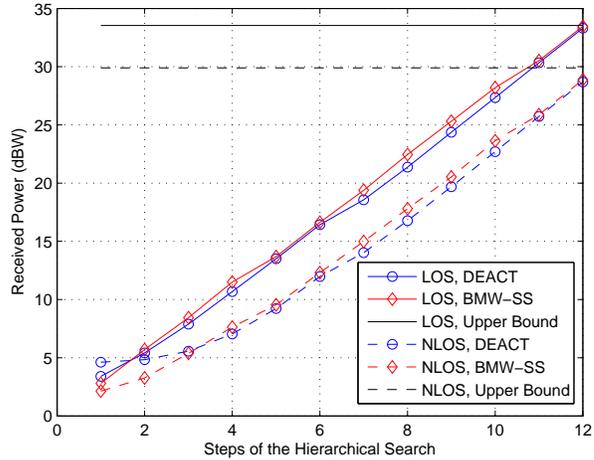}
  \caption{Received power during each search step with the BMW-SS and DEACT codebooks under both LOS and NLOS channels, where $N_{\rm{T}}=N_{\rm{R}}=64$, $L=3$, $P_{\rm{tot}}=1$ W, and $N_0=10^{-4}$ W. Step 1 to Step 6 is for Rx training, while Step 7 to Step 12 is for Tx training.}
  \label{fig:received_power}
\end{center}
\end{figure}

Fig. \ref{fig:Success_Rate_LOS} shows the success rate of hierarchical search with the BMW-SS, DEACT and Sparse (proposed in \cite{alkhateeb2014channel}) codebooks under LOS channel, where $N_{\rm{T}}=N_{\rm{R}}=64$, $L=3$. $\eta$ is the power difference in dB between the LOS component and an NLOS MPC. From this figure, it is observed that both the transmission SNR $\gamma_{\rm{tot}}$ and $\eta$ affect the success rate. For all the codebooks, the success rate improves as $\gamma_{\rm{tot}}$ increases. However, due to the mutual effect of MPCs (i.e., spatial fading), the success rate improves little when $\gamma_{\rm{tot}}$ is already high enough. Basically when $\eta$ is bigger, the mutual effect is less, and the success rate is higher. For the Sparse codebook, the performance also depends on the number of RF chains. When the number of RF chains is small, e.g., only 2, the deep sinks within the covered angle (See Fig. \ref{fig:beam_Cmp}) will sharply reduce the success rate, as shown in Fig. \ref{fig:Success_Rate_LOS}. Furthermore, we can find that the success rate with the BMW-SS codebook is higher than that with the DEACT codebook. This is because the beams of the BMW-SS codebook are flatter than those of the DEACT codebook; thus they are more robust to the spatial fading. Also, the success rate with the BMW-SS codebook is higher than that with the Sparse codebook when the number of RF chains is not large.

\begin{figure}[t]
\begin{center}
  \includegraphics[width=\figwidth cm]{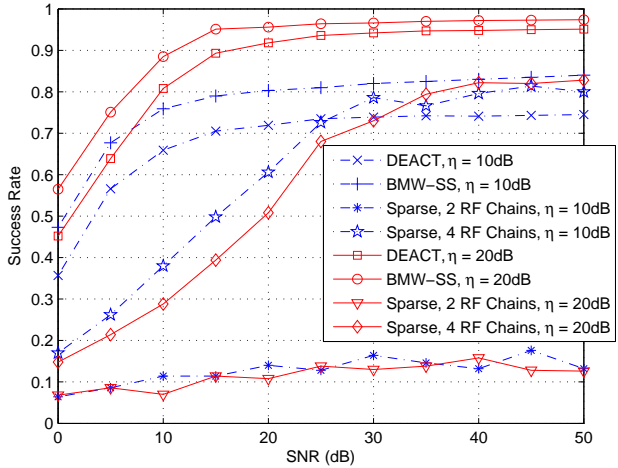}
  \caption{Success rate of hierarchical search with the BMW-SS, DEACT and Sparse codebooks under LOS channel, where $N_{\rm{T}}=N_{\rm{R}}=64$, $L=3$. $\eta$ is the power difference in dB between the LOS component and an NLOS MPC.}
  \label{fig:Success_Rate_LOS}
\end{center}
\end{figure}

Fig. \ref{fig:Success_Rate_NLOS} shows the success rate of hierarchical search with the BMW-SS, DEACT and Sparse codebooks under NLOS channel, where $N_{\rm{T}}=N_{\rm{R}}=64$. From this figure, the same performance variation with respect to the transmission SNR $\gamma_{\rm{tot}}$ can be observed as that in Fig. \ref{fig:Success_Rate_LOS}, and Sparse with 2 RF chains also has the poorest performance. In addition, BMW-SS basically outperforms DEACT and Sparse with 4 RF chains, and the superiority depends on $L$. When $L=1$, i.e., there is only one MPC, both BMW-SS and DEACT achieve a 100\% success rate when $\gamma_{\rm{tot}}$ is high enough, because there is no spatial fading. In contrast, Sparse cannot achieve a 100\% success rate even when $L=1$, due to the deep sinks within the covered angle. When $L>1$, all these schemes can hardly achieve a 100\% success rate, due to the mutual effect of multiple MPCs.

\begin{figure}[t]
\begin{center}
  \includegraphics[width=\figwidth cm]{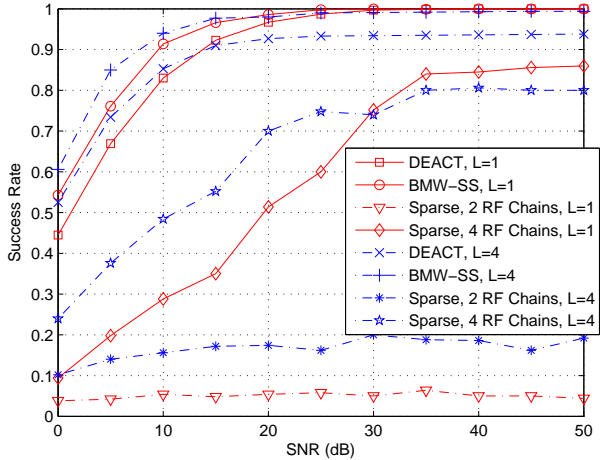}
  \caption{Success rate of hierarchical search with the BMW-SS, DEACT and Sparse codebooks under NLOS channel, where $N_{\rm{T}}=N_{\rm{R}}=64$.}
  \label{fig:Success_Rate_NLOS}
\end{center}
\end{figure}

It is noteworthy that the superiority of BMW-SS versus DEACT in Fig. \ref{fig:received_power} is different from that in Figs. \ref{fig:Success_Rate_LOS} and \ref{fig:Success_Rate_NLOS}. Fig. \ref{fig:received_power} just shows the variation of the received power during the search process with a high SNR, while Figs. \ref{fig:Success_Rate_LOS} and \ref{fig:Success_Rate_NLOS} show the search results over a wide SNR range. Fig. \ref{fig:received_power} actually corresponds to the search process for a set of points in Figs. \ref{fig:Success_Rate_LOS} and \ref{fig:Success_Rate_NLOS} with SNR equal of 40dB, and the received-power superiority of BMW-SS over DEACT in Fig. \ref{fig:received_power} will become larger if smaller SNR values are adopted.

\subsection{Per-Antenna Transmission Power Model}
In this subsection, the per-antenna transmission power signal model shown in \eqref{eq_signal_perp} is used to compare the limit performances of BMW-SS and DEACT with the same per-antenna transmission power. With this model, the deactivation of antennas will significantly affect the total transmission power. In particular, the total transmission power is lower if the number of active antennas is smaller.

Fig. \ref{fig:received_power_per} shows the received power during each search step with the BMW-SS and DEACT codebooks under both LOS and NLOS channels, where $N_{\rm{T}}=N_{\rm{R}}=64$, $L=3$, $P_{\rm{per}}=1$ W, and $N_0=10^{-4}$ W. The upper bound is achieved by the exhaustive search method. For the LOS channel, the LOS component has 15dB higher power than that of an NLOS MPC. Comparing this figure with Fig. \ref{fig:received_power}, we find a significant difference that with the per-antenna transmission power model BMW-SS has a distinct superiority over DEACT during the search process, especially at the beginning of the search process. The superiority is about 15 dB at the beginning, and it becomes less as the search goes on, until vanishes at the end of beam search, i.e., the two methods achieve the same received SNR after the search process. The superiority of BMW-SS results from the fact that the number of the active antennas for the codewords with wide beams is significantly greater than that for DEACT, and thus BMW-SS has a much higher total transmission power than DEACT when the per-antenna transmission power is the same.

 Moreover, the increasing speed of received power is the same from Step 1 to Step 6 for both of the two schemes, but from Step 7 to Step 12, the increasing speed for BMW-SS varies, and that for DEACT becomes greater than that from Step 1 to Step 6. This is because with per-antenna transmission power, there are two power gains during the search process according to \eqref{eq_power_gain_per}, namely the array gain provided by narrowing the Tx/Rx beams and the total transmission power gain provided by increasing the number of active transmit antennas. For DEACT, there is only Rx array gain from Step 1 to Step 6, where Rx training is performed, while there are both Tx array gain and total transmission power gain from Step 7 to Step 12, where Tx training is performed; thus, the increasing speed of received power is greater from Step 7 to Step 12. For BMW-SS, there is also only Rx array gain from Step 1 to Step 6 for Rx training; thus the received power consistently increases with the same speed as DEACT. But from Step 7 to Step 12 for Tx training, although the Tx beam consistently becomes narrower, which means that Tx array gain is consistently improved, the number of active antennas alternatively changes between $N_{\rm{T}}$ and $N_{\rm{T}}/2$, which means that the total transmission power may become larger or smaller. Hence, when both the Tx array gain and total transmission power increase, the received power improves with a speed the same as DEACT, while when the Tx array gain increases but the total transmission power decreases, the received SNR does not improve and may even decrease.

It is noted that the superiority of BMW-SS over DEACT at the beginning of the search process is with big significance for mmWave communication, where per-antenna transmission power is generally limited. This superiority guarantees that with the BMW-SS codebook, the success rate of beam search will be upgraded with the same transmission distance, or the transmission distance will be extended with the same success rate of beam search.

\begin{figure}[t]
\begin{center}
  \includegraphics[width=\figwidth cm]{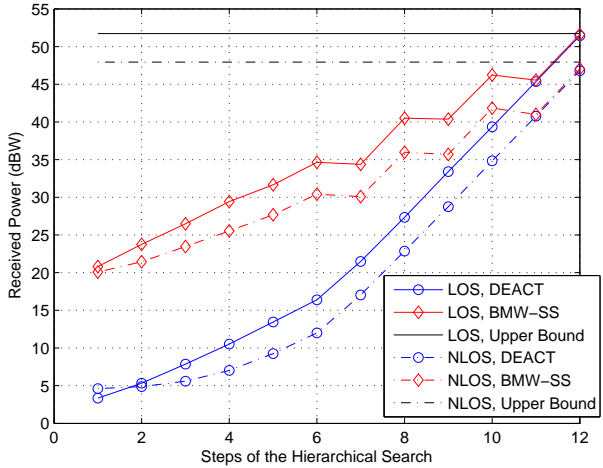}
  \caption{Received SNR during each search step with the BMW-SS and DEACT codebooks under both LOS and NLOS channels, where $N_{\rm{T}}=N_{\rm{R}}=64$, $L=3$, $P_{\rm{per}}=1$ W, and $N_0=10^{-4}$ W. Step 1 to Step 6 is for Rx training, while Step 7 to Step 12 is for Tx training.}
  \label{fig:received_power_per}
\end{center}
\end{figure}

Figs. \ref{fig:Success_Rate_LOS_per} and \ref{fig:Success_Rate_NLOS_per} show the success rates of hierarchical search with the BMW-SS and DEACT codebooks under LOS and NLOS channels, respectively. The same simulation conditions are adopted as those in Figs. \ref{fig:Success_Rate_LOS} and \ref{fig:Success_Rate_NLOS}, respectively, and the same results can be obtained from Figs. \ref{fig:Success_Rate_LOS_per} and \ref{fig:Success_Rate_NLOS_per} as those from Figs. \ref{fig:Success_Rate_LOS} and \ref{fig:Success_Rate_NLOS}, respectively, except that the superiority of BMW-SS over DEACT becomes more significant in Figs. \ref{fig:Success_Rate_LOS_per} and \ref{fig:Success_Rate_NLOS_per}, which benefits from not only the fact that the beams of the BMW-SS codebook are flatter than those of the DEACT codebook, but also that the number of the active antennas of the BMW-SS codewords is basically much greater than that of DEACT, which offers much higher total transmission power. Also, Figs. \ref{fig:Success_Rate_LOS_per} and \ref{fig:Success_Rate_NLOS_per} reveal that even with low per-antenna transmission power, the success rate of BMW-SS can be close to 100\%, which is evidently better than those of DEACT and Sparse.

\begin{figure}[t]
\begin{center}
  \includegraphics[width=\figwidth cm]{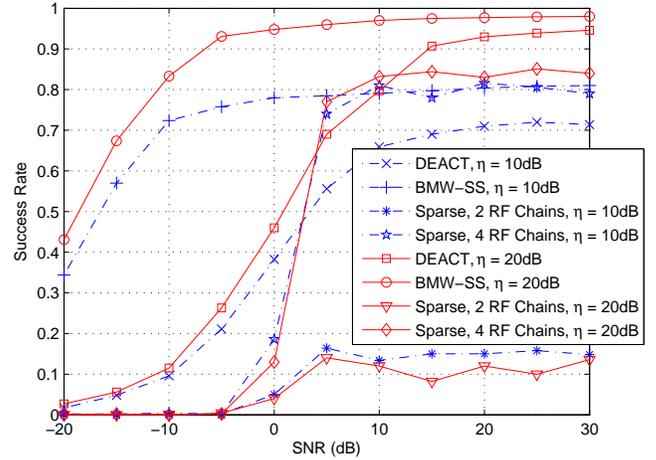}
  \caption{Success rate of hierarchical search with the BMW-SS and DEACT codebooks under LOS channel, where $N_{\rm{T}}=N_{\rm{R}}=64$, $L=3$. $\eta$ is the power difference in dB between the LOS component and an NLOS MPC.}
  \label{fig:Success_Rate_LOS_per}
\end{center}
\end{figure}

\begin{figure}[t]
\begin{center}
  \includegraphics[width=\figwidth cm]{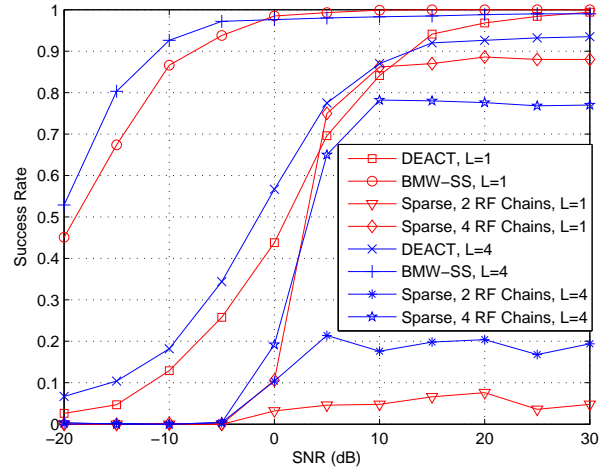}
  \caption{Success rate of hierarchical search with the BMW-SS and DEACT codebooks under NLOS channel, where $N_{\rm{T}}=N_{\rm{R}}=64$.}
  \label{fig:Success_Rate_NLOS_per}
\end{center}
\end{figure}

\section{Conclusion}
In this paper hierarchical codebook design has been studied for mmWave communication. Firstly, two basic criteria have been proposed for the codebook design. Next, a complete binary-tree structured hierarchical codebook has been designed by jointly using sub-array and deactivation techniques, i.e., the BMW-SS approach. Performance evaluation has been conducted with both a total transmission power and a per-antenna transmission power system models. Results show that the BMW-SS codebook offers two advantages over the deactivation codebook, namely flatter beams and much larger number of active antennas. Both of these two advantages basically provide performance superiorities in terms of received power during the search process and the success rate of beam search under both transmission power models, and the performance superiority is especially significant with the per-antenna transmission power system model. In addition, the BMW-SS codebook also outperforms the Sparse codebook, since there are no deep sinks within the beam coverage for BMW-SS.

\appendices

\section{Proof of Theorem 1}
Given an arbitrary $N$-element vector ${\bf{w}}$ and two arbitrary angles $\psi$ and $\Omega$, we want to prove that ${\cal C}{\cal V}({\bf{w}} \circ \sqrt N {\bf{a}}(N,\psi ))= {\cal C}{\cal V}({\bf{w}}) + \psi$, where $\mathcal{A}+\psi$ is a new angle set with elements being those of the angle set $\mathcal{A}$ added by $\psi$. Note that ${\bf{w}} \circ \sqrt N {\bf{a}}(N,\psi )$ is actually a new vector achieved based on ${\bf{w}}$ and the steering vector ${\bf{a}}(N,\psi )$. Let us first see the beam gain of this new vector.

\begin{equation} \label{eq_new_vector}
\begin{aligned}
&A({\bf{w}} \circ \sqrt N {\bf{a}}(N,\psi ),\Omega )\\
\mathop {=}\limits^{(a)}& \sqrt N {\bf{a}}{(N,\Omega )^{\rm{H}}}({\bf{w}} \circ \sqrt N {\bf{a}}(N,\psi ))\\
\mathop {=}\limits^{(b)}& \sum\limits_{n = 1}^N {{{[{\bf{w}}]}_n}{e^{{\rm{j}}\pi (n - 1)\psi }}{e^{ - {\rm{j}}\pi (n - 1)\Omega }}} \\
 =& \sum\limits_{n = 1}^N {{{[{\bf{w}}]}_n}{e^{ - {\rm{j}}\pi (n - 1)(\Omega  - \psi )}}} \\
\mathop {=}\limits^{(c)}& A({\bf{w}},\Omega  - \psi ),
\end{aligned}
\end{equation}
where (a) and (c) are according to the definition of the beam gain in \eqref{eq_beam_gain}, while (b) is obtained according to definition of the entry-wise product.

Thus, we further have
\begin{equation}
\begin{aligned}
&{\cal C}{\cal V}({\bf{w}} \circ \sqrt N {\bf{a}}(N,\psi ))\\
\mathop {=}\limits^{(a)}& \{ \Omega |\;|A({\bf{w}} \circ \sqrt N {\bf{a}}(N,\psi ),\Omega )| >\\
 &~~~~~~~~~~~~~ \rho \mathop {\max }\limits_\omega  |A({\bf{w}} \circ \sqrt N {\bf{a}}(N,\psi ),\omega )|\} \\
\mathop {=}\limits^{(b)}& \{ \Omega |\;|A({\bf{w}},\Omega  - \psi )| > \rho \mathop {\max }\limits_\omega  |A({\bf{w}},\omega  - \psi )|\} \\
\mathop {=}\limits^{(c)}& \{ \Omega |\;|A({\bf{w}},\Omega  - \psi )| > \rho \mathop {\max }\limits_\omega  |A({\bf{w}},\omega )|\} \\
\mathop {=}\limits^{(d)}& \{ {\Omega _0} + \psi |\;|A({\bf{w}},{\Omega _0})| > \rho \mathop {\max }\limits_\omega  |A({\bf{w}},\omega )|\} \\
\mathop {=}\limits^{(e)}& \{ {\Omega _0}|\;|A({\bf{w}},{\Omega _0})| > \rho \mathop {\max }\limits_\omega  |A({\bf{w}},\omega )|\}  + \psi \\
 =& {\cal C}{\cal V}({\bf{w}}) + \psi,
\end{aligned}
\end{equation}
where (a) is according to the definition of beam coverage in \eqref{eq_beam_coverage}, (b) is according to \eqref{eq_new_vector}, (c) is based on the fact that the maxima of $|A({\bf{w}},\Omega  - \psi )|$ does not depend on the angle offset $\psi$, (d) is obtained by letting $\Omega={\Omega _0} + \psi$, and (e) is obtained according to the definition of an angle set plus a single angle in Theorem 1.

\section*{Acknowledgment}

The authors would like to thank the editor and the anonymous reviewers for
their many useful and detailed comments that have helped to improve the
presentation of this manuscript. The authors would also like to thank the authors of \cite{alkhateeb2014channel} to share their source code online, and particularly thank Dr. Ahmed Alkhateeb for his kind help to explain how to use the source code.


\begin{thebibliography}{1}

\bibitem{Rapp_2010_60GHz_general}
R.~C. Daniels, J.~N. Murdock, T.~S. Rappaport, and R.~W. Heath, ``60 {GHz}
  wireless: up close and personal,'' \emph{IEEE Microwave Magazine}, vol.~11,
  no.~7, pp. 44--50, Dec. 2010.

\bibitem{wang_2011_MMWCS}
K.-C. Huang and Z.~Wang, \emph{Millimeter Wave Communication Systems}.\hskip
  1em plus 0.5em minus 0.4em\relax Hoboken, New Jersey, USA: Wiley-IEEE Press,
  2011.

\bibitem{Park_2010_11ad}
E.~Perahia, C.~Cordeiro, M.~Park, and L.~L. Yang, ``{IEEE} 802.11 ad: defining
  the next generation multi-{Gbps} {Wi-Fi},'' in \emph{IEEE Consumer
  Communications and Networking Conference (CCNC)}.\hskip 1em plus 0.5em minus
  0.4em\relax Las Vegas, NV: IEEE, Jan. 2010, pp. 1--5.

\bibitem{Xia_2011_60GHz_Tech}
S.~K. Yong, P.~Xia, and A.~Valdes-Garcia, \emph{{60GHz} Technology for {Gbps}
  {WLAN} and {WPAN}: from Theory to Practice}.\hskip 1em plus 0.5em minus
  0.4em\relax West Sussex, UK: Wiley.

\bibitem{xiaozhenyu2013div}
Z.~Xiao, ``Suboptimal spatial diversity scheme for 60 {GHz} millimeter-wave
  {WLAN},'' \emph{IEEE Communications Letters}, vol.~17, no.~9, pp. 1790--1793,
  Sept. 2013.

\bibitem{xia_2008_prac_ante_traning}
P.~Xia, H.~Niu, J.~Oh, and C.~Ngo, ``Practical antenna training for millimeter
  wave {MIMO} communication,'' in \emph{IEEE Vehicular Technology Conference
  (VTC) 2008}.\hskip 1em plus 0.5em minus 0.4em\relax Calgary, Canada: IEEE,
  Oct. 2008, pp. 1--5.

\bibitem{xia_2008_multi_stage}
P.~Xia, S.~K. Yong, J.~Oh, and C.~Ngo, ``Multi-stage iterative antenna training
  for millimeter wave communications,'' in \emph{IEEE GLOBECOM Conference
  2008}.\hskip 1em plus 0.5em minus 0.4em\relax New Orleans, LA, USA: IEEE,
  Dec. 2008, pp. 1--6.

\bibitem{khan_2011}
F.~Khan and J.~Pi, ``Millimeter-wave mobile broadband: unleashing 3--300{GHz}
  spectrum,'' in \emph{IEEE Wireless Commun. Netw. Conf.}, Cancun, Mexico,
  March 2011.

\bibitem{alkhateeb2014mimo}
A.~Alkhateeb, J.~Mo, N.~Gonz{\'a}lez-Prelcic, and R.~Heath, ``{MIMO} precoding
  and combining solutions for millimeter-wave systems,'' \emph{IEEE
  Communications Magazine}, vol.~52, no.~12, pp. 122--131, Dec. 2014.

\bibitem{choi2014coding}
J.~Choi, ``On coding and beamforming for large antenna arrays in mm-wave
  systems,'' \emph{IEEE Wireless Communications Letters}, vol.~3, no.~2, pp.
  193--196, April 2014.

\bibitem{han2015large}
S.~Han, I.~Chih-Lin, Z.~Xu, and C.~Rowell, ``Large-scale antenna systems with
  hybrid analog and digital beamforming for millimeter wave {5G},'' \emph{IEEE
  Communications Magazine}, vol.~53, no.~1, pp. 186--194, Jan. 2015.

\bibitem{roh2014millimeter}
W.~Roh, J.-Y. Seol, J.~Park, B.~Lee, J.~Lee, Y.~Kim, J.~Cho, K.~Cheun, and
  F.~Aryanfar, ``Millimeter-wave beamforming as an enabling technology for {5G}
  cellular communications: theoretical feasibility and prototype results,''
  \emph{IEEE Communications Magazine}, vol.~52, no.~2, pp. 106--113, Feb. 2014.

\bibitem{sun2014mimo}
S.~Sun, T.~S. Rappaport, R.~Heath, A.~Nix, and S.~Rangan, ``{MIMO} for
  millimeter-wave wireless communications: beamforming, spatial multiplexing,
  or both?'' \emph{IEEE Communications Magazine}, vol.~52, no.~12, pp.
  110--121, Dec. 2014.

\bibitem{niu2015survey}
Y.~Niu, Y.~Li, D.~Jin, L.~Su, and A.~V. Vasilakos, ``A survey of millimeter
  wave communications (mmwave) for {5G}: opportunities and challenges,''
  \emph{Wireless Networks}, pp. 1--20, April 2015.

\bibitem{wang2014tens}
P.~Wang, Y.~Li, X.~Yuan, L.~Song, and B.~Vucetic, ``Tens of gigabits wireless
  communications over e-band los {MIMO} channels with uniform linear antenna
  arrays,'' \emph{IEEE Transactions on Wireless Communications}, vol.~13,
  no.~7, pp. 3791--3805, July 2014.

\bibitem{wang2015multi}
P.~Wang, Y.~Li, L.~Song, and B.~Vucetic, ``Multi-gigabit millimeter wave
  wireless communications for 5g: from fixed access to cellular networks,''
  \emph{IEEE Communications Magazine}, vol.~53, no.~1, pp. 168--178, Jan. 2015.

\bibitem{wang_2009_beam_codebook}
J.~Wang, Z.~Lan, C.~Pyo, T.~Baykas, C.~Sum, M.~Rahman, J.~Gao, R.~Funada,
  F.~Kojima, and H.~Harada, ``Beam codebook based beamforming protocol for
  multi-{G}bps millimeter-wave {WPAN} systems,'' \emph{IEEE Journal on Selected
  Areas in Communications}, vol.~27, no.~8, pp. 1390--1399, Oct. 2009.

\bibitem{wang_2009_beam_codebook_vtc}
J.~Wang, Z.~Lan, C.~Sum, C.~Pyo, J.~Gao, T.~Baykas, A.~Rahman, R.~Funada,
  F.~Kojima, and I.~Lakkis, ``Beamforming codebook design and performance
  evaluation for {60GHz} wideband {WPANs},'' in \emph{IEEE Vehicular Technology
  Conference Fall (VTC 2009-Fall)}.\hskip 1em plus 0.5em minus 0.4em\relax
  Anchorage, AK: IEEE, Sept. 2009, pp. 1--6.

\bibitem{alkhateeb2014channel}
A.~Alkhateeb, O.~El~Ayach, G.~Leus, and R.~Heath, ``Channel estimation and
  hybrid precoding for millimeter wave cellular systems,'' \emph{IEEE Journal
  of Selected Topics in Signal Processing}, vol.~8, no.~5, pp. 831--846, Oct.
  2014.

\bibitem{alkhateeb2015compressed}
A.~Alkhateeb, G.~Leus, and R.~W. Heath~Jr, ``Compressed sensing based
  multi-user millimeter wave systems: How many measurements are needed?''
  \emph{arXiv preprint arXiv:1505.00299}, May 2015.

\bibitem{peng2015enhanced}
Y.~Peng, Y.~Li, and P.~Wang, ``An enhanced channel estimation method for
  millimeter wave systems with massive antenna arrays,'' \emph{IEEE
  Communications Letters}, vol.~19, no.~9, pp. 1592--1595, Sept. 2015.

\bibitem{Xiaozy2014BeamTrain}
Z.~Xiao, L.~Bai, and J.~Choi, ``Iterative joint beamforming training with
  constant-amplitude phased arrays in millimeter-wave communications,''
  \emph{IEEE Communications Letters}, vol.~18, no.~5, pp. 829--832, May 2014.

\bibitem{kokshoorn2015fast}
M.~Kokshoorn, P.~Wang, Y.~Li, and B.~Vucetic, ``Fast channel estimation for
  millimetre wave wireless systems using overlapped beam patterns,'' in
  \emph{IEEE International Conference on Communications (ICC)}.\hskip 1em plus
  0.5em minus 0.4em\relax London, UK: IEEE, June 2015, pp. 1304--1309.

\bibitem{TseFundaWC}
D.~Tse and P.~Viswanath, \emph{Fundamentals of Wireless Communication}.\hskip
  1em plus 0.5em minus 0.4em\relax New York, USA: Cambridge University Press.

\bibitem{chen2011multi}
L.~Chen, Y.~Yang, X.~Chen, and W.~Wang, ``Multi-stage beamforming codebook for
  {60GHz} {WPAN},'' in \emph{2011 6th International ICST Conference on
  Communications and Networking in China (CHINACOM)}.\hskip 1em plus 0.5em
  minus 0.4em\relax Harbin, China: IEEE, Aug. 2011, pp. 361--365.

\bibitem{hur2013millimeter}
S.~Hur, T.~Kim, D.~J. Love, J.~V. Krogmeier, T.~A. Thomas, and A.~Ghosh,
  ``Millimeter wave beamforming for wireless backhaul and access in small cell
  networks,'' \emph{IEEE Transactions on Communications}, vol.~61, no.~10, pp.
  4391--4403, Oct. 2013.

\bibitem{he2015suboptimal}
T.~He and Z.~Xiao, ``Suboptimal beam search algorithm and codebook design for
  millimeter-wave communications,'' \emph{Mobile Networks and Applications},
  pp. 86--97, Jan. 2015.

\bibitem{xiaozhenyuGC2013}
Z.~Xiao, X.-G. Xia, D.~Jin, and N.~Ge, ``Multipath grouping for millimeter-wave
  communications,'' in \emph{IEEE Global Communications Conference (GLOBECOM)},
  Atlanta, GA, Dec 2013, pp. 3378--3383.

\bibitem{xiao2015Iterative}
------, ``Iterative eigenvalue decomposition and multipath-grouping {Tx/Rx}
  joint beamformings for millimeter-wave communications,'' \emph{IEEE
  Transactions on Wireless Communications}, vol.~14, no.~3, pp. 1595--1607,
  March 2015.

\bibitem{rapp_2013_MMW}
T.~Rappaport, F.~Gutierrez, E.~Ben-Dor, J.~Murdock, Y.~Qiao, and J.~Tamir,
  ``Broadband millimeter wave propagation measurements and models using
  adaptive beam antennas for outdoor urban cellular communications,''
  \emph{IEEE Transactions on Antennas and Propagation}, vol.~61, no.~4, pp.
  1850--1859, April 2013.

\bibitem{Rapp_2012_cellular_MMW}
T.~S. Rappaport, Y.~Qiao, J.~I. Tamir, J.~N. Murdock, and E.~Ben-Dor,
  ``Cellular broadband millimeter wave propagation and angle of arrival for
  adaptive beam steering systems,'' in \emph{IEEE Radio and Wireless Symposium
  (RWS)}.\hskip 1em plus 0.5em minus 0.4em\relax Santa Clara, CA: IEEE, Jan.
  2012, pp. 151--154.

\bibitem{sayeed_2007}
A.~M. Sayeed and V.~Raghavan, ``Maximizing {MIMO} capacity in sparse multipath
  with reconfigurable antenna arrays,'' \emph{IEEE Journal of Selected Topics
  in Signal Processing}, vol.~1, no.~1, pp. 156--166, June 2007.

\bibitem{Ayach2014}
O.~El~Ayach, S.~Rajagopal, S.~Abu-Surra, Z.~Pi, and R.~Heath, ``Spatially
  sparse precoding in millimeter wave {MIMO} systems,'' \emph{IEEE Transactions
  on Wireless Communications}, vol.~13, no.~3, pp. 1499--1513, March 2014.

\bibitem{nsenga_2009}
J.~Nsenga, W.~Van~Thillo, F.~Horlin, V.~Ramon, A.~Bourdoux, and R.~Lauwereins,
  ``Joint transmit and receive analog beamforming in 60 {GHz MIMO} multipath
  channels,'' in \emph{IEEE International Conference on Communications
  (ICC)}.\hskip 1em plus 0.5em minus 0.4em\relax Dresden, Germany: IEEE, June
  2009, pp. 1--5.

\bibitem{libin2013}
B.~Li, Z.~Zhou, W.~Zou, X.~Sun, and G.~Du, ``On the efficient beam-forming
  training for {60GHz} wireless personal area networks,'' \emph{IEEE
  Transactions on Wireless Communications}, vol.~12, no.~2, pp. 504--515, Feb.
  2013.

\end{thebibliography}
\end{document}